\documentclass[11pt]{article}

\usepackage[utf8]{inputenc}
\usepackage[T1]{fontenc}
\usepackage[margin=1in]{geometry}
\usepackage{amsmath,amssymb,amsthm,mathtools}
\usepackage{mathrsfs}
\usepackage{booktabs}
\usepackage{array}
\usepackage{enumitem}
\usepackage{hyperref}
\usepackage{xcolor}
\hypersetup{
  colorlinks=true,
  linkcolor=red,
  citecolor=red,
  urlcolor=red
}

\theoremstyle{plain}
\newtheorem{theorem}{Theorem}

\theoremstyle{definition}

\theoremstyle{remark}

\newcommand{\R}{\mathbb{R}}
\newcommand{\Z}{\mathbb{Z}}
\newcommand{\C}{\mathbb{C}}
\newcommand{\HH}{\mathbb{H}}
\newcommand{\sltwo}{\mathfrak{sl}(2,\R)}
\newcommand{\SL}{\mathrm{SL}}
\newcommand{\SO}{\mathrm{SO}}
\newcommand{\Spin}{\mathrm{Spin}}
\newcommand{\SU}{\mathrm{SU}}
\newcommand{\Mp}{\mathrm{Mp}}
\newcommand{\Kt}{\widetilde{K}}
\newcommand{\Gt}{\widetilde{G}}
\newcommand{\SLt}{\widetilde{\mathrm{SL}}}
\newcommand{\Sigmah}{\widehat{\Sigma}}
\newcommand{\OO}{\mathcal{O}}
\newcommand{\OOh}{\widehat{\mathcal{O}}}
\newcommand{\D}{\mathcal{D}}
\newcommand{\DD}{\mathbb{D}}
\newcommand{\Hil}{\mathcal{H}}
\DeclareMathOperator{\tr}{tr}

\newcommand{\ii}{\mathrm{i}}

\title{Symmetry Regularization of 1D Generalized Coulomb Problems}
\author{Zhanqiang Bai\thanks{School of Mathematical Sciences, Soochow University, Suzhou 215006, China; zqbai@suda.edu.cn}
\and Junwei Ma\thanks{Department of Mathematics, Hong Kong University of Science and Technology, Clear Water Bay, Hong Kong, China; jmaas@connect.ust.hk}
\and Guowu Meng\thanks{Department of Mathematics, Hong Kong University of Science and Technology, Clear Water Bay, Hong Kong, China; mameng@ust.hk}}
\date{June 6, 2026}

\begin{document}

\maketitle

\begin{abstract}
For the 1D generalized Coulomb problems---a family that includes the quantizations of the 1D generalized Kepler problems of Ma--Meng--Xiao~\cite{MMX2025}---we construct two explicit unitary intertwiners $\hat\iota_{\pm}$, the quantum analogs of the classical symmetry regularization maps $\iota_{\pm}$ of~\cite{MMX2025}, that unitarily identify each $H_{\kappa}$-energy-definite portion of the Hilbert space $L^{2}(\R_{>0},\mathrm{d}q)$ with a unitary lowest-weight representation of $\SLt(2,\R)$.
\end{abstract}



\tableofcontents

\section{Introduction}

For the classical Kepler problem, collision motions blow up in finite time;
since the 1960s, various \emph{regularization methods}~\cite{KustaanheimoStiefel1965,Moser1970} have been devised to
complete the dynamics on a suitable enlargement of phase space. Considering
that symmetry is more fundamental than dynamics, the third author proposed
in~\cite{Meng2023} an alternative notion, \emph{symmetry regularization},
which shifts the criterion from dynamical completeness to symmetry
completeness: the hidden symmetry of each energy-definite portion of the phase space should extend to a global group action on a manifold containing that portion as a dense open subset.

A concrete setting is provided by the \emph{1D generalized Kepler problems},
a family of systems on the phase space $\Sigma:=T^{*}\R_{>0}$ parametrized by
a magnetic charge $\mu\ge 0$, with symplectic form
$\mathrm{d}p\wedge \mathrm{d}q$ and
Hamiltonian\footnote{The symbol $H$ is used in this paper for the classical Hamiltonian
\eqref{eq:Hcl} and its quantum analog $H_{\kappa}$ of~\eqref{eq:Hkappa}; the same
letter also denotes the standard non-compact Cartan generator of $\sltwo$ introduced
in Section~\ref{ssec:four-models} (see~\eqref{eq:epm} and the surrounding text).
The intended meaning will always be clear from context.}
\begin{equation}\label{eq:Hcl}
  H \;=\; \tfrac12 p^{2}\;+\;\frac{\mu^{2}}{2q^{2}}\;-\;\frac{1}{q},
  \qquad (q,p)\in\R_{>0}\times\R,
\end{equation}
originally introduced by Meng in~\cite[Ch.~15]{Meng2024Oziewicz} with a
slightly different normalization. When $\mu>0$ the centrifugal barrier
$\mu^{2}/(2q^{2})$ already prevents collisions at $q=0$, but for $\mu=0$
collisions can occur in finite time and the system requires regularization.
In~\cite{MMX2025}, Ma--Meng--Xiao construct a pair of symplectic embeddings
\[
  \iota_{-}\colon\Sigma_{-}\longrightarrow\OO_{-}:=\OO_{\mu},\qquad
  \iota_{+}\colon\Sigma_{+}\longrightarrow\OO_{+}:=\OO_{\mu/2},
\]
of the positive- and negative-energy subsets $\Sigma_{\pm}:=\{\pm H>0\}$ into coadjoint orbits of
$\sltwo$ (elliptic for $\mu>0$, regular nilpotent for $\mu=0$). Each $\iota_{\pm}$ has dense open image and intertwines the Kepler
flow with a \emph{complete} Hamiltonian flow on $\OO_{\pm}$, thereby achieving both dynamical and symmetry completeness.

The purpose of the present article is to work out the \textbf{quantum
counterpart} of this picture. The \emph{1D generalized Coulomb problems},
first introduced in~\cite{Meng2014}, are a one-parameter family of
Schr\"odinger operators $H_{\kappa}$ on $L^{2}(\R_{>0},\mathrm{d}q)$ indexed
by $\kappa>0$ that includes the quantizations of the above classical systems
as special cases. (We adopt
the term ``Coulomb'' for the quantum side because of the appearance of the Coulomb potential $-1/q$
in the Hamiltonian operator $H_{\kappa}$.)
The quantum analogs of $\iota_{\pm}$ unitarily identify each
$H_{\kappa}$-energy-definite portion of the Hilbert space
$L^{2}(\R_{>0},\mathrm{d}q)$ with a unitary lowest-weight representation
of $\SLt(2,\R)$. The
overall picture matches the classical one row-by-row:
\begin{center}
\renewcommand{\arraystretch}{1.25}
\begin{tabular}{@{}ll@{}}
\toprule
Classical & Quantum \\
\midrule
$\Sigma_{-},\;\Sigma_{+}$
  & $\Sigmah_{-},\;\Sigmah_{+}$\\
$\OO_{-}=\OO_{\mu}$
  & $\OOh_{-}=\D^{+}_{\kappa}$\\
$\OO_{+}=\OO_{\mu/2}$
  & $\OOh_{+}=\D^{+}_{(\kappa+1)/2}$\\
$\iota_{-}\colon\Sigma_{-}\!\to\!\OO_{-}$
  & $\hat\iota_{-}\colon\Sigmah_{-}\!\to\!\OOh_{-}$ \\
$\iota_{+}\colon\Sigma_{+}\!\to\!\OO_{+}$
  & $\hat\iota_{+}\colon\Sigmah_{+}\!\to\!\OOh_{+}$\\
\bottomrule
\end{tabular}
\end{center}
\noindent The quantum column is valid for all $\kappa>0$; the classical
counterpart exists only when $\kappa\ge 1$ (i.e.\ $\mu\ge 0$, where
$\kappa=2\mu+1$). See~Section~\ref{sec:Hkappa} for the full convention. Since quantization is not a functor, the classical column serves as motivation rather than a recipe: passing from the classical to the quantum column requires genuinely new constructions, not a straightforward application of a quantization procedure.

The organization of this article is as follows.
Section~\ref{sec:coadjoint} collects the relevant facts about the
elliptic coadjoint orbits $\OO_{\mu}$ of $\sltwo$, their identification with the
upper half-plane $\HH$, and the geometric quantization that produces the
holomorphic discrete series $\D^{+}_{2\mu+1}$.
Section~\ref{sec:classical-reg} reviews the classical
symmetry regularization maps $\iota_{\pm}$ of~\cite{MMX2025} and explains why
\emph{S-duality} is the appropriate name for them.
Section~\ref{sec:discrete} recalls the representation theory needed on the quantum side: the universal
cover $\SLt(2,\R)$, the unitary lowest-weight representations
$\D^{+}_{\kappa}$ for $\kappa>0$, four concrete models (two half-line, two
Bergman), and, in the range $\kappa\ge 1$, their interpretation via the orbit
method.
Section~\ref{sec:Hkappa} defines the 1D generalized Coulomb operator $H_{\kappa}$,
derives its scattering and bound-state (generalized) eigenfunctions
via the Kummer equation,
and introduces the subspaces $\Sigmah_{\pm}$ as the quantum
analogs of $\Sigma_{\pm}$.
Section~\ref{sec:quantum-reg} states and proves the main result
(Theorem~\ref{thm:iotahat}): for each sign of the energy, a unitary
isomorphism $\hat\iota_{\pm}\colon\Sigmah_{\pm}\to\OOh_{\pm}$ is
constructed by matching the spectral data of $|2H_{\kappa}|^{-1/2}$ on
$\Sigmah_{\pm}$ with that of the corresponding $\sltwo$-generator on the
target representation.
We conclude in Section~\ref{sec:concluding} with remarks on
higher-dimensional generalizations and connections with the work of Fock
and Bander--Itzykson.

\section{Coadjoint orbits of \texorpdfstring{$\sltwo$}{sl(2,R)}}\label{sec:coadjoint}

This section describes the elliptic coadjoint orbits $\OO_{\mu}$ of $\mathfrak{g}:=\sltwo$, their identification with the Poincar\'e
disk and the upper half-plane, and the geometric quantization (with half-form
correction) that produces the holomorphic discrete series
$\D^{+}_{2\mu+1}$ of $\SLt(2,\R)$.

On $\mathfrak g$ we use the invariant pseudo-Euclidean form
$(X,Y):=-2\tr(XY)$ and the basis
\[
  v_{1}=\tfrac12\begin{pmatrix}1&0\\0&-1\end{pmatrix},\quad
  v_{2}=\tfrac12\begin{pmatrix}0&1\\1&0\end{pmatrix},\quad
  v_{0}=\tfrac12\begin{pmatrix}0&1\\-1&0\end{pmatrix},
\]
with $(v_{1},v_{1})=(v_{2},v_{2})=-1$ and $(v_{0},v_{0})=+1$.
Pairing with $v_{i}$ yields the coordinate function $x_{i}$ on $\mathfrak{g}^{*}$,
so this identifies $\mathfrak{g}^{*}\cong\R^{1,2}$.

For $\mu\ge 0$ set
\[
  \OO_{\mu}\;:=\;\bigl\{\,x\in\R^{1,2}\ \big|\ x_{0}^{2}-x_{1}^{2}-x_{2}^{2}=\mu^{2},\;x_{0}>0\,\bigr\}.
\]
For $\mu=0$ this is the future light cone (regular nilpotent orbit); for $\mu>0$
it is the upper sheet of a two-sheeted hyperboloid (elliptic orbit).

The Kirillov--Kostant--Souriau symplectic form on $\OO_{\mu}$
(see~\cite[Ch.~1]{Kirillov2004} for the general construction) is
\[
  \omega\;=\;\frac{\mathrm{d}x_{1}\wedge \mathrm{d}x_{2}}{x_{0}}.
\]
This is the \emph{opposite} sign convention from~\cite{MMX2025}, where
$\omega_{\text{MMX}}=\mathrm{d}x_{2}\wedge\mathrm{d}x_{1}/x_{0}=-\omega$;
our choice here ensures that $\omega$ pulls back to a positive multiple of the
hyperbolic area form on~$\HH$ (see below). 

\subsection{Identification of \texorpdfstring{$\OO_{\mu}$}{O\_mu} (\texorpdfstring{$\mu>0$}{mu>0}) with the upper half-plane
  and geometric quantization}

For $\mu>0$, let $G:=\SL(2,\R)$ act on $\OO_{\mu}$ via the coadjoint action and let
$K:=\SO(2)\subset G$ be the stabilizer of $(\mu,0,0)\in\OO_{\mu}$.
The $G$-equivariant identification of $\OO_{\mu}\cong G/K$ with the
Poincar\'e disk $\DD=\{w\in\C:|w|<1\}$ sends a point
$(x_{0},x_{1},x_{2})\in\OO_{\mu}$ to
\[
  w\;=\;\frac{x_{1}+\ii\,x_{2}}{x_{0}+\mu}\;\in\;\DD,
\]
mapping the apex $(\mu,0,0)$ to the origin and intertwining the coadjoint
action with the $G$-action on $\DD$.
(Equivalently, this is the stereographic projection from $(-\mu,0,0)$,
rescaled by $1/\mu$ so that the image is the \emph{unit} disk rather than
a disk of radius~$\mu$.)
Composing with the Cayley map
\[
  z\;=\;\ii\,\frac{1+w}{1-w},
\]
the standard biholomorphism from $\DD$ to the upper
half-plane $\HH=\{z\in\C:\Im z>0\}$, gives an identification
$\OO_{\mu}\xrightarrow{\;\sim\;}\HH$.

Under the composite identification $\OO_{\mu}\cong\HH$, the
Kirillov--Kostant--Souriau form becomes
\[
  \omega\;\longmapsto\;
  \omega_{\mu}:=\mu\,\frac{\mathrm{d}x\wedge\mathrm{d}y}{y^{2}},
  \qquad z=x+\ii y\in\HH.
\]
That is, the KKS form on $\OO_{\mu}$ is identified with $\mu$ times the standard
$\SL(2,\R)$-invariant area form on $\HH$.

Equip $\HH$ with the symplectic form $\omega_{\mu}$
and the K\"ahler polarization given by its complex structure.
The prequantum line bundle $L_{\mu}\to\HH$ is the Hermitian line bundle
with curvature $\omega_{\mu}$. Trivializing sections as holomorphic functions
on $\HH$, the $\SLt(2,\R)$-action carries M\"obius weight~$2\mu$:
\[
  (g\cdot f)(z)
  =(cz+d)^{-2\mu}\,f\!\left(\frac{az+b}{cz+d}\right),\qquad
  g^{-1}=\begin{pmatrix}a&b\\c&d\end{pmatrix},
\]
and the natural $L^{2}$-condition on sections of $L_{\mu}$ gives the
weighted Bergman space $\Hil^{\HH}_{2\mu}$
(squared norm $\|f\|^{2}=\int_{\HH}|f|^{2}\,y^{2\mu-2}\,\mathrm{d}x\,\mathrm{d}y$).
The canonical bundle $\mathcal{K}=T^{*(1,0)}\HH$ of this polarization has sections $g(z)\,\mathrm{d}z$
with weight~$2$ (since $\mathrm{d}((az+b)/(cz+d))
=(cz+d)^{-2}\,\mathrm{d}z$), so the half-form bundle $\mathcal{K}^{1/2}$
carries weight~$1$.
Geometric quantization with the half-form (metaplectic) correction
(see~\cite{Sniatycki1980,Woodhouse1991} for textbook
treatments) replaces $L_{\mu}$ by $L_{\mu}\otimes \mathcal{K}^{1/2}$, shifting
the M\"obius weight from $2\mu$ to $2\mu+1$. The resulting representation is the
\emph{holomorphic discrete series representation}
$\D_{2\mu+1}^{+}$ of $\SLt(2,\R)$ (made precise in Section~\ref{ssec:disc-series}
below), with Hilbert space
\[
  \Hil^{\HH}_{2\mu+1}\;=\;\Bigl\{\,f\in\mathcal{O}(\HH)\;\Big|\;
  \int_{\HH}|f(z)|^{2}\,y^{2\mu-1}\,\mathrm{d}x\,\mathrm{d}y<\infty\Bigr\},
\]
and the $\SLt(2,\R)$-action is the weight-$(2\mu+1)$ M\"obius action lifted to
the universal cover.
Equivalently, in the Poincar\'e disk model $\DD$ (via the Cayley map
$w=(z-\ii)/(z+\ii)$), the Hilbert space is
\[
  \Hil^{\DD}_{2\mu+1}\;=\;\Bigl\{\,g\in\mathcal{O}(\DD)\;\Big|\;
  \int_{\DD}|g(w)|^{2}\,(1-|w|^{2})^{2\mu-1}\,\mathrm{d}u\,\mathrm{d}v<\infty\Bigr\},
  \qquad w=u+\ii v.
\]

As $\mu\to 0^{+}$, the elliptic orbits $\OO_{\mu}$ degenerate to the regular nilpotent orbit $\OO_{0}$ and the symplectic form
$\omega_{\mu}$ degenerates to $0$. On the quantum side, the holomorphic discrete series
representations $\D_{2\mu+1}^{+}$ correspondingly degenerate to the
\emph{limit of holomorphic discrete series} $\D_{1}^{+}$. This provides a natural
\emph{quantization of the nilpotent orbit $\OO_{0}$}: one first quantizes the nearby elliptic orbits $\OO_{\mu}$
to obtain the holomorphic discrete series, then passes to the limit $\mu\to 0$.

\section{Classical symmetry regularization: the maps
  \texorpdfstring{$\iota_{\pm}:\Sigma_{\pm}\to\OO_{\pm}$}{iota±:Σ±→O±}}\label{sec:classical-reg}

We now recall the two embeddings of Ma--Meng--Xiao~\cite{MMX2025}. Both are symplectic embeddings with dense open image in their respective target orbits. Throughout
this section, $\OO_{-}:=\OO_{\mu}$ and $\OO_{+}:=\OO_{\mu/2}$, and $H$ denotes
the classical Hamiltonian~\eqref{eq:Hcl}.

\subsection{The negative-energy embedding
  \texorpdfstring{$\iota_{-}:\Sigma_{-}\to\OO_{-}$}{iota\_-}}

Define first $\tilde\iota_{-}:\Sigma_{-}\to\OO_{\mu}$ by
\begin{equation}\label{eq:tildeiotaminus}
  \xi_{0}\;=\;\frac{1}{\sqrt{-2H}},\qquad
  \xi_{1}\;=\;-\,\frac{1+2Hq}{\sqrt{-2H}},\qquad
  \xi_{2}\;=\;qp,
\end{equation}
a diffeomorphism onto the dense open set
$U^{-}_{\mu}=\{(x_{0},x_{1},x_{2})\in\OO_{\mu}\colon x_{0}+x_{1}\neq 0\}$
(equal to all of $\OO_{\mu}$ when $\mu>0$). Then applying an elliptic twist $\tau_{-}$ by
$\varphi=\xi_{2}/\xi_{0}=\sqrt{-2H}\,qp$ gives:
\begin{equation}\label{eq:iotaminus}
  \iota_{-}\;=\;\bigl(\xi_{0},\;\xi_{1}\cos\varphi-\xi_{2}\sin\varphi,\;\xi_{1}\sin\varphi+\xi_{2}\cos\varphi\bigr).
\end{equation}

The map $\iota_{-}$ satisfies $\iota_{-}^{*}\omega_{\text{MMX}}=\mathrm{d}p\wedge \mathrm{d}q$ (symplectic
embedding,~\cite[Prop.~1]{MMX2025}); equivalently,
$\iota_{-}^{*}\omega=-\mathrm{d}p\wedge \mathrm{d}q=\mathrm{d}q\wedge \mathrm{d}p$ in our sign convention.
The map is one-to-one because of a lemma of Ligon--Schaaf~\cite{LigonSchaaf1976}
(proof given also in~\cite{MMX2025}): $\theta=A\sin(\theta+\gamma)$ has a unique real
solution when $|A|\le 1$. It is in fact an embedding onto a dense open
subset of $\OO_{-}$, and pulls back the orbit Hamiltonian $h_{-}:=1/(-2x_{0}^{2})$
to $H$:
\begin{equation}\label{eq:pullbackminus}
  \boxed{\;\iota_{-}^{*}h_{-}\;=\;H.\;}
\end{equation}
Since the energy surfaces $\{h_{-}=E\}$ are compact subsets of $\OO_{-}$, the
flow is complete on $\OO_{-}$ \emph{even when $\mu=0$}, where the original
Kepler flow may blow up in finite time. This is the dynamical completion; it is also a symmetry completion, since $\OO_{-}$ carries the full $G$-action.

\subsection{The positive-energy embedding
  \texorpdfstring{$\iota_{+}:\Sigma_{+}\to\OO_{+}$}{iota\_+}}

In the positive case, $\iota_{+}$ is the composition of three maps. First,
define $\tilde\iota_{+}\colon\Sigma_{+}\to\OO_{\mu}$ by
\begin{equation}\label{eq:tildeiotaplus}
  \xi_{0}\;=\;\frac{1+2Hq}{\sqrt{2H}},\qquad
  \xi_{1}\;=\;qp,\qquad
  \xi_{2}\;=\;\frac{1}{\sqrt{2H}}.
\end{equation}
Then applying a hyperbolic twist $\tau_{+}$ by $\varphi=\xi_{1}/\xi_{2}=\sqrt{2H}\,qp$ gives:
\begin{equation}\label{eq:tauiotaplus} \tau_{+}\tilde\iota_{+}\;=\;\bigl(\xi_{0}\cosh\varphi-\xi_{1}\sinh\varphi,\;-\xi_{0}\sinh\varphi+\xi_{1}\cosh\varphi,\;\xi_{2}\bigr).
\end{equation}
Finally, apply the squaring map
\begin{equation}\label{eq:squaring}
  \pi\colon (x_{0},x_{1},x_{2})\;\longmapsto\;
  \Bigl(\tfrac{x_{0}}{2},\;\tfrac{x_{1}^{2}-x_{2}^{2}}{2\sqrt{x_{1}^{2}+x_{2}^{2}}},\;\tfrac{x_{1}x_{2}}{\sqrt{x_{1}^{2}+x_{2}^{2}}}\Bigr).
\end{equation}
The composition $\iota_{+}=\pi\,\tau_{+}\tilde\iota_{+}$ lands in
$\OO_{+}=\OO_{\mu/2}$, is symplectic (with respect to
$\omega_{\text{MMX}}$), and is a smooth embedding with dense
open image (a proper subset of $\OO_{+}$). The orbit Hamiltonian
\[
  h_{+}\;=\;\frac{1/4}{x_{1}^{2}+x_{2}^{2}-x_{1}\sqrt{x_{1}^{2}+x_{2}^{2}}}
\]
(whose denominator vanishes on the ray $\{x_{2}=0,\,x_{1}>0\}$, the
locus excluded from the image of $\iota_{+}$)
is pulled back to $H$:
\begin{equation}\label{eq:pullbackplus}
  \boxed{\;\iota_{+}^{*}h_{+}\;=\;H.\;}
\end{equation}
This is the dynamical completion (see~\cite{MMX2025} for the detailed argument); it is also a symmetry completion, since $\OO_{+}$ carries the full $G$-action.

The role of the squaring map $\pi\colon\OO_{\mu}\to\OO_{\mu/2}$ is essential here: without it, the image
of $\tau_{+}\tilde\iota_{+}$ would fill only one half of $\OO_{\mu}$
(the $x_{2}>0$ half), which is far from being an open dense subset.
The squaring map $\pi$, a $2{:}1$ ramified cover, identifies the two halves
so that the composite $\iota_{+}$ has dense open image in $\OO_{\mu/2}$. Writing $z=x_{1}+\ii x_{2}$, the
last two components of $\pi(x_{0},x_{1},x_{2})$ combine to $z^{2}/(2|z|)$; thus
the orbit parameter is halved from $\mu$ to $\mu/2$, and the Casimir
$x_{0}^{2}-x_{1}^{2}-x_{2}^{2}$ drops from $\mu^{2}$ to $\mu^{2}/4$
as a side effect of the squaring. \emph{This halving of the orbit parameter ($\mu\mapsto\mu/2$) on the positive
energy side will reappear at the quantum level as the parameter shift
$\kappa\mapsto(\kappa+1)/2$.}

\subsection{Why ``S-duality'' is the right name}

As emphasized in \cite{MMX2025}, the term ``regularization map'' is a misnomer
when $\mu>0$: $\Sigma_{\pm}$ is already dynamically complete and $\iota_{-}$
is a diffeomorphism onto $\OO_{\mu}$, so no (generalized) Kepler motion needs to be regularized.
A more accurate name is \emph{S-duality}, in the sense familiar from
gauge and string theory: a change of variables that exchanges a
description in which a symmetry is hidden for one in which it is
manifest. As explained in~\cite{MMX2025}, the maps $\iota_{\pm}$ share
two hallmarks of S-duality: (i)~they turn hidden symmetries into manifest
ones, and (ii)~they are implemented by a Fourier-type transform between
conjugate variables. The dynamical completion that occurs at $\mu=0$ is
then a byproduct, not the defining feature, of the S-duality map.

From the quantum vantage point this is even clearer: quantization of
a symplectic manifold $M$ and of a dense open submanifold $M'\subset M$
yields the \emph{same} Hilbert space (the quantum states are insensitive to
sets of measure zero). Whatever dynamical incompleteness $M'$ may suffer is
therefore invisible to the quantization, and the name ``regularization'' --- in
the sense of Moser~\cite{Moser1970} and Ligon--Schaaf~\cite{LigonSchaaf1976} ---
is unmistakably a misnomer from the quantum point of view.

\section{Unitary lowest-weight representations of \texorpdfstring{$\SLt(2,\R)$}{SL\textasciitilde(2,R)}}\label{sec:discrete}

This section introduces the unitary lowest-weight representations
$\D^{+}_{\kappa}$ ($\kappa>0$) of the universal cover $\SLt(2,\R)$,
which generalize the holomorphic discrete series obtained by quantizing the coadjoint
orbits of Section~\ref{sec:coadjoint} to include the limit of holomorphic discrete
series ($\kappa=1$) and the analytic-continuation family ($0<\kappa<1$).
We recall four unitarily equivalent concrete models --- the half-line
(Kirillov) and Whittaker models, and the Bergman models on $\HH$ and $\DD$
--- which will reappear in the spectral analysis of later sections.
Standard references for the material recalled here are
\cite{HoweTan1992} (which the reader new to $\SL(2,\R)$ representation
theory may consult throughout) and~\cite{Sally1967} for the analytic
continuation in~$\kappa$.

\subsection{The universal cover}

Since $K=\SO(2)\cong S^{1}$ is a deformation retract of $G=\SL(2,\R)$,
$\pi_{1}(G)=\pi_{1}(K)=\Z$. The universal cover
\[
  \Gt\;=\;\SLt(2,\R)\quad\text{fits in the short exact sequence}\quad
  1\to\Z\to\Gt\to\SL(2,\R)\to 1.
\]
The Iwasawa decomposition lifts to $\Gt=\Kt AN$, with the non-compact
factors $A,N$ unchanged and $\Kt\cong\R$ the universal cover of $K=\SO(2)$.
The characters of $\Kt$ are $\chi_{\kappa}(\theta)=e^{\ii\kappa\theta}$ for
\emph{arbitrary} $\kappa\in\R$ --- not just integers. This is the source of
the continuous family of representations on $\Gt$ that have no analog on
$\SL(2,\R)$.

\subsection{The unitary lowest-weight representations \texorpdfstring{$\D^{+}_{\kappa}$}{D+\_kappa}}
\label{ssec:disc-series}

For every real $\kappa>0$, the unitary lowest-weight representation
$\D^{+}_{\kappa}$ of $\Gt$\footnote{We follow the representation-theoretic convention in which the subscript~$\kappa$ equals the lowest $\Kt$-type.} (see~\cite[Ch.~II]{HoweTan1992}
for a concrete construction and~\cite{Sally1967} for the analytic
continuation in~$\kappa$) has $\Kt$-types
\[
  \{\kappa,\;\kappa+2,\;\kappa+4,\;\ldots\},
\]
each with multiplicity one, and a lowest-weight vector annihilated by the
lowering operator (see~\eqref{eq:epm} below).
\begin{itemize}[leftmargin=2em,topsep=2pt,itemsep=2pt]
  \item For $\kappa\in\Z_{\ge 2}$, $\D^{+}_{\kappa}$ descends to the holomorphic
        discrete series $D^{+}_{\kappa}$ of $\SL(2,\R)$;
  \item for $\kappa=1$ it is the limit of holomorphic discrete series $D^{+}_{1}$;
  \item for $\kappa\in\tfrac12+\Z_{\ge 0}$,
        $\D^{+}_{\kappa}$ descends to the metaplectic group $\Mp(2,\R)$ but not to $\SL(2,\R)$
        (the Weil representation is
        $\D^{+}_{1/2}\oplus\D^{+}_{3/2}$; see~\cite[Ch.~IV]{HoweTan1992});
  \item for generic real $\kappa>0$ it lives only on $\Gt$.
\end{itemize}
The Casimir element of $\sltwo$ is defined only up to a nonzero
multiplicative scalar; we fix the normalization so that it acts on
$\D^{+}_{\kappa}$ by $\Omega:=\kappa(\kappa-2)/4$
(positive for $\kappa>2$, zero at $\kappa=2$, negative for $0<\kappa<2$). The representation is
square-integrable for $\kappa>1$, tempered but non-square-integrable for $\kappa=1$,
and unitary but non-tempered for
$0<\kappa<1$ (obtained by analytic continuation in~$\kappa$~\cite{Sally1967}).

\subsection{Four concrete models}\label{ssec:four-models}

We use the standard basis of $\sltwo$,
\[
  H=\begin{pmatrix}1&0\\0&-1\end{pmatrix},\qquad
  E=\begin{pmatrix}0&1\\0&0\end{pmatrix},\footnote{The symbol $E$ is also used
  later in the paper to denote the energy (i.e.\ the eigenvalue of the
  Hamiltonian $H_{\kappa}$), for instance in Section~\ref{ssec:spectral-decomp}
  where $E=k^{2}/2$; the intended meaning will be clear from context.}\qquad
  F=\begin{pmatrix}0&0\\1&0\end{pmatrix},
\]
with $[H,E]=2E$, $[H,F]=-2F$, $[E,F]=H$ (see the introduction's
footnote on the overloaded symbol $H$). In terms of the basis
$\{v_{0},v_{1},v_{2}\}$ of Section~\ref{sec:coadjoint}, one has $H=2v_{1}$, $E=v_{2}+v_{0}$,
$F=v_{2}-v_{0}$. We also write $h:=\ii(F-E)=-2\ii v_{0}$ for the elliptic Cartan
generator and
\begin{equation}\label{eq:epm}
  e_{\pm}=\tfrac12(H\pm \ii(E+F)),\qquad [h,e_{\pm}]=\pm 2e_{\pm},
\end{equation}
so $e_{+}$ raises and $e_{-}$ lowers the $\Kt$-type.
All four of the following models will be used below; standard
references for the equivalences and intertwiners between them are
\cite{HoweTan1992,Sally1967}. In what follows, $\Hil_{\kappa}$ (no superscript)
denotes the abstract Hilbert space of $\D^{+}_{\kappa}$;
$\Hil^{\mathrm{Kir}}_{\kappa}$, $\Hil^{\mathrm{Wh}}_{\kappa}$, $\Hil^{\HH}_{\kappa}$,
$\Hil^{\DD}_{\kappa}$ are its concrete realizations in the four models below.
The two half-line models~(i)--(ii) are valid for all $\kappa>0$; the two
Bergman models~(iii)--(iv) are concrete $L^{2}$-spaces only for $\kappa>1$;
for $\kappa\le 1$ they are defined by reproducing-kernel completion
(see the remark after model~(iv) below).

\medskip\noindent\textbf{(i) Half-line (Kirillov) model.}
Hilbert space $\Hil^{\mathrm{Kir}}_{\kappa}=L^{2}(\R_{>0},\,t^{\kappa-1}\,\mathrm{d}t)$. The Lie
algebra acts by
\[
  \pi_{\kappa}(H)=2t\partial_{t}+\kappa,\qquad
  \pi_{\kappa}(E)=\ii t,\qquad
  \pi_{\kappa}(F)=\ii\bigl(t\partial_{t}^{2}+\kappa\partial_{t}\bigr).
\]
The (normalized) lowest-weight vector is $\phi_{0}(t)=\sqrt{2^{\kappa}/\Gamma(\kappa)}\,e^{-t}$ (annihilated by $e_{-}$),
and applying the raising operator $e_{+}$ produces the orthonormal basis
\[
  \phi_{n}(t)=\sqrt{\frac{n!\,2^{\kappa}}{\Gamma(n+\kappa)}}\;L_{n}^{(\kappa-1)}(2t)\,e^{-t},\qquad n=0,1,2,\ldots,
\]
of $\Kt$-type $\kappa+2n$. Here $L_{n}^{(\alpha)}$ denotes the generalized
(associated) Laguerre polynomial of degree $n$ and parameter $\alpha$, i.e.\ the
orthogonal polynomials on $(0,\infty)$ with weight $x^{\alpha}e^{-x}$
(see~\cite[Ch.~18]{DLMF}); they appear here because successive applications of
the raising operator $e_{+}$ to the Gaussian-type lowest-weight vector
$\phi_{0}\propto e^{-t}$ generate precisely $L_{n}^{(\kappa-1)}(2t)\,e^{-t}$.

\medskip\noindent\textbf{(ii) Half-line (Whittaker) model.}
Hilbert space $\Hil^{\mathrm{Wh}}_{\kappa}=L^{2}(\R_{>0},\,t^{-1}\,\mathrm{d}t)$, obtained
from (i) by the unitary conjugation $f\mapsto t^{\kappa/2}f$. Writing
$\hat\pi_{\kappa}(X):=t^{\kappa/2}\,\pi_{\kappa}(X)\,t^{-\kappa/2}$ for the
conjugated representation, the Lie algebra acts by
\[
  \hat\pi_{\kappa}(H)=2t\partial_{t},\qquad
  \hat\pi_{\kappa}(E)=\ii t,\qquad
  \hat\pi_{\kappa}(F)=\ii\!\left(t\partial_{t}^{2}-\frac{\kappa(\kappa-2)}{4t}\right).
\]
The elliptic Cartan
generator $h=\ii(F-E)$ becomes a Schr\"odinger-type operator
\begin{equation}\label{eq:hatpih-Schrodinger}
  \hat\pi_{\kappa}(h)\;=\;-t\partial_{t}^{2}\;+\;\frac{\kappa(\kappa-2)}{4t}\;+\;t.
\end{equation}
Its complete orthonormal eigenbasis $\{\hat\phi_{n}\}_{n=0}^{\infty}$, with $\hat\phi_{n}:=t^{\kappa/2}\phi_{n}$ the image of the Kirillov basis under conjugation by $t^{\kappa/2}$, satisfies
\begin{equation}\label{eq:hatphin-eigen}
  \hat\pi_{\kappa}(h)\,\hat\phi_{n}\;=\;(\kappa+2n)\,\hat\phi_{n},\qquad n=0,1,2,\ldots.
\end{equation}

\medskip\noindent\textbf{(iii) Bergman (upper-half-plane) model.}
Hilbert space $\Hil^{\HH}_{\kappa}$ of holomorphic functions on $\HH=\{z=x+\ii y:y>0\}$
with $\|f\|^{2}_{\kappa}=\int_{\HH}|f|^{2}\,y^{\kappa-2}\,\mathrm{d}x\,\mathrm{d}y<\infty$, with
the cocycle action
\[
  (\pi_{\kappa}(g)f)(z)=(cz+d)^{-\kappa}f\!\left(\frac{az+b}{cz+d}\right),\qquad
  g^{-1}=\begin{pmatrix}a&b\\c&d\end{pmatrix}.
\]
For non-integer $\kappa$, single-valuedness of $(cz+d)^{-\kappa}$ requires lifting
to $\Gt$.

\medskip\noindent\textbf{(iv) Bergman (Poincar\'e disk) model.}
Hilbert space $\Hil^{\DD}_{\kappa}$ of holomorphic functions on the unit disk
$\DD=\{w\in\C:|w|<1\}$ with
$\|f\|^{2}_{\DD,\kappa}=\int_{\DD}|f(w)|^{2}(1-|w|^{2})^{\kappa-2}\,\mathrm{d}^{2}w<\infty$.
The Cayley transform $w=(z-\ii)/(z+\ii)$ identifies $\HH$ with $\DD$ and intertwines
(iii) with (iv); under this identification the M\"obius action becomes
\[
  (\pi_{\kappa}(g)f)(w)=(\bar\beta w+\bar\alpha)^{-\kappa}\,
  f\!\left(\frac{\alpha w+\beta}{\bar\beta w+\bar\alpha}\right),\qquad
  g^{-1}=\begin{pmatrix}\alpha&\beta\\\bar\beta&\bar\alpha\end{pmatrix}\in\SU(1,1),
\]
with $|\alpha|^{2}-|\beta|^{2}=1$. In this model the elliptic Cartan generator $h$
is diagonalized by the monomial basis $\{w^{n}\}_{n\ge 0}$, with $w^{n}$ of
$\Kt$-type $\kappa+2n$; the lowest-weight vector is the constant function $1$.

Strictly, the weights $(1-|w|^{2})^{\kappa-2}$ in (iv) and $y^{\kappa-2}$ in (iii)
are locally integrable only for $\kappa>1$; for $0<\kappa\le 1$ both Bergman models
are understood as reproducing-kernel completions, obtained by analytic
continuation in~$\kappa$~\cite{Sally1967}.

\section{The 1D generalized Coulomb problems}\label{sec:Hkappa}

Throughout this and the remaining sections we use the
representation-theoretic parameter $\kappa>0$. For $\kappa\ge 1$ it is
related to the classical magnetic charge $\mu\ge 0$ by
$\kappa=2\mu+1$; the range $0<\kappa<1$ has no classical counterpart
but is perfectly meaningful on the quantum side, since $\D^{+}_{\kappa}$
and $H_{\kappa}$ are well-defined for every $\kappa>0$.
In terms of $\kappa$, the Casimir element (in the normalization of Section~\ref{ssec:disc-series}) acts on $\D^{+}_{\kappa}$
by the scalar $\Omega=\kappa(\kappa-2)/4$,
and the classical squaring $\OO_{\mu}\to\OO_{\mu/2}$ (for $\mu\ge 0$)
corresponds to the representation-parameter shift
$\kappa\mapsto(\kappa+1)/2$ on the quantum side.

Since quantization is not a functor, the classical picture of
Sections~\ref{sec:coadjoint} and~\ref{sec:classical-reg} serves as a guide rather than a recipe.  Following
\cite{Meng2014}, the 1D generalized Coulomb problems are introduced
directly from the unitary lowest-weight representations, via the TKK algebra.
The \emph{universal Kepler problem} of \cite{Meng2014} is an abstract
object attached to a Euclidean Jordan algebra $V$, encoded by its
TKK algebra. A \emph{suitable} Poisson realization
of this TKK algebra produces a (classical) generalized Kepler problem;
likewise, a \emph{suitable} operator realization of this TKK algebra
produces a generalized quantum Kepler problem --- i.e.\ a generalized
Coulomb problem. For the simplest Jordan algebra $V=\R$ the conformal algebra is
$\sltwo$, and the suitable operator realizations of the
TKK algebra are precisely the unitary lowest-weight
representations $\D^{+}_{\kappa}$. Each $\D^{+}_{\kappa}$ yields a
Hamiltonian in the following $\kappa>0$ family:
\begin{equation}\label{eq:Hkappa}
  \boxed{\;
    H_{\kappa}\;=\;-\,\frac12\,\frac{\mathrm{d}^{2}}{\mathrm{d}q^{2}}\;+\;\frac{\kappa(\kappa-2)}{8q^{2}}\;-\;\frac{1}{q}.
  \;}
\end{equation}
The inverse-square coefficient $\kappa(\kappa-2)/8$ equals $\Omega/2$; it is
repulsive for $\kappa>2$, vanishes at $\kappa=0,2$, and is
attractive for $0<\kappa<2$.

Viewed as a symmetric operator in $L^{2}(\R_{>0},\mathrm{d}q)$ with domain
$C_{c}^{\infty}(\R_{>0})$,
\eqref{eq:Hkappa} is essentially self-adjoint only for $\kappa\ge 3$.
For $0<\kappa<3$ the two Frobenius solutions of $H_{\kappa}\psi=E\psi$ at $q=0$, behaving as
$q^{\kappa/2}$ and $q^{1-\kappa/2}$, are both square-integrable near the
origin, so a choice of boundary condition is needed. For
$1\le\kappa<3$, the canonical choice is the \emph{Friedrichs extension}
(equivalently, requiring $\psi(q)=O(q^{\kappa/2})$ as $q\to 0^{+}$);
for $0<\kappa<1$, it is the \emph{Krein (anti-Friedrichs) extension}
(equivalently, requiring $\psi(q)=O(q^{1-\kappa/2})$ as $q\to 0^{+}$).
In both cases, the chosen extension is the unique one whose bound-state
subspace carries $\D^{+}_{\kappa}$ representation-theoretically. \emph{Henceforth $H_{\kappa}$ denotes this
canonical self-adjoint operator.}

The assignment $\D^{+}_{\kappa}\longleftrightarrow H_{\kappa}$ is a bijection
\begin{equation*}
  \boxed{\;\begin{gathered}
    \bigl\{\text{unitary lowest-weight representations of }\SLt(2,\R)\bigr\}\\
    \updownarrow\;{1{:}1}\\
    \bigl\{\text{1D generalized Coulomb problems}\bigr\}.
  \end{gathered}\;}
\end{equation*}

\subsection{Spectral decomposition}\label{ssec:spectral-decomp}

$H_{\kappa}$ has the standard Coulomb-type spectral structure:
\[
  \sigma(H_{\kappa})\;=\;\{E_{n}\}_{n=0}^{\infty}\;\cup\;[0,\infty),\qquad
  E_{n}=-\frac{1/2}{(n+\kappa/2)^{2}}\nearrow 0.
\]
Correspondingly,
\[
  L^{2}(\R_{>0},\mathrm{d}q)\;=\;\Sigmah_{-}\,\oplus\,\Sigmah_{+},\footnote{The threshold energy $E=0$ lies in the spectrum but carries zero spectral measure, so it does not contribute to either subspace.}
\]
where $\Sigmah_{-}$ (resp.\ $\Sigmah_{+}$) is the closed subspace on which
$H_{\kappa}$ is negative-definite (resp.\ positive-definite). These are the natural
quantum counterparts of the classical $\Sigma_{\pm}=\{\pm H>0\}$.
We now derive the explicit (generalized) eigenfunctions of $H_{\kappa}$ that span each subspace.

\subsubsection*{The Kummer equation}

The \textbf{Kummer equation} (or confluent hypergeometric equation) is
\begin{equation}\label{eq:Kummer-ODE}
  z\,w''(z)\;+\;(b-z)\,w'(z)\;-\;a\,w(z)\;=\;0,
\end{equation}
with parameters $a,b\in\C$.  The solution regular at the origin is
\textbf{Kummer's function of the first kind},
\begin{equation}\label{eq:1F1-def}
  {}_{1}F_{1}(a;\,b;\,z)\;=\;\sum_{n=0}^{\infty}\frac{(a)_{n}}{(b)_{n}}\,\frac{z^{n}}{n!},
\end{equation}
where $(a)_{n}=\Gamma(a+n)/\Gamma(a)$ is the Pochhammer symbol
(see~\cite[Ch.~13]{DLMF} and~\cite{WhittakerWatson1927}).

\subsubsection*{Scattering eigenfunctions (continuous spectrum)}

The generalized eigenvalue equation $H_{\kappa}\psi=\frac{k^{2}}{2}\psi$ at positive
energy $E=k^{2}/2>0$ reads
\begin{equation}\label{eq:Coulomb-ODE}
  \psi''(q)\;+\;\left(k^{2}\;+\;\frac{2}{q}\;-\;\frac{\kappa(\kappa-2)/4}{q^{2}}\right)\psi(q)\;=\;0.
\end{equation}
Substituting $\psi(q)=q^{\kappa/2}\,e^{-\ii kq}\,f(2\ii kq)$ and writing
$z=2\ii kq$, a direct calculation reduces~\eqref{eq:Coulomb-ODE} to the
Kummer equation~\eqref{eq:Kummer-ODE} with parameters
\[
  a\;=\;\frac{\kappa}{2}+\frac{\ii}{k},\qquad b\;=\;\kappa.
\]
The solution regular at the origin is $f(z)={}_{1}F_{1}(a;\,b;\,z)$.
The normalized scattering generalized eigenfunction of $H_{\kappa}$ is
\begin{equation}\label{eq:psi-E-Coulomb}
  \psi_{E}(q)\;=\;\mathcal{C}_{\kappa}(k)\,q^{\kappa/2}\,e^{-\ii kq}\;
    {}_{1}F_{1}\!\bigl(\tfrac{\kappa}{2}+\tfrac{\ii}{k};\,\kappa;\,2\ii kq\bigr),
\end{equation}
with the Coulomb normalization constant
\begin{equation}\label{eq:Coulomb-norm}
  \mathcal{C}_{\kappa}(k)\;=\;\frac{1}{\sqrt{2\pi}}\,(2k)^{\kappa/2}\,
    e^{\pi/(2k)}\,\frac{\bigl|\Gamma(\tfrac{\kappa}{2}+\tfrac{\ii}{k})\bigr|}{\Gamma(\kappa)}.
\end{equation}
These satisfy~\cite[Ch.~33]{DLMF}
$\int_{0}^{\infty}\overline{\psi_{E(k)}(q)}\,\psi_{E(k')}(q)\,\mathrm{d}q
=\delta(k-k')$, and the \textbf{scattering subspace} $\Sigmah_{+}$ is
the closed span of $\{\psi_{E}\}_{E>0}$.

\subsubsection*{Bound-state eigenfunctions (discrete spectrum)}

The eigenvalue equation $H_{\kappa}\psi=-\frac{1}{2N_{n}^{2}}\psi$ at negative
energy $E_{n}=-1/(2N_{n}^{2})$, with $N_{n}:=n+\kappa/2$, is obtained from
the same ODE~\eqref{eq:Coulomb-ODE} by the substitution $k\mapsto -\ii/N_{n}$.
Applying the same substitution to the scattering ansatz gives
$\psi(q)=q^{\kappa/2}\,e^{-q/N_{n}}\,f(2q/N_{n})$; writing
$z=2q/N_{n}$, one arrives at the Kummer equation~\eqref{eq:Kummer-ODE}
with
\[
  a\;=\;-n,\qquad b\;=\;\kappa.
\]
Since $a=-n$ is a non-positive integer, the Kummer function terminates and
yields a generalized \textbf{Laguerre polynomial}:
\begin{equation}\label{eq:Laguerre-1F1}
  L_{n}^{(\alpha)}(z)\;=\;\binom{n+\alpha}{n}\;
    {}_{1}F_{1}(-n;\;\alpha+1;\;z),
  \qquad n=0,1,2,\dotsc
\end{equation}
The normalized bound-state eigenfunction is therefore
\begin{equation}\label{eq:psi-n-Laguerre}
  \psi_{n}(q)\;=\;\mathcal{N}_{n}\,q^{\kappa/2}\,e^{-q/N_{n}}\,
    L_{n}^{(\kappa-1)}(2q/N_{n}),\qquad
  \mathcal{N}_{n}\;=\;\frac{1}{N_{n}}\Bigl(\frac{2}{N_{n}}\Bigr)^{\!\kappa/2}
    \sqrt{\frac{n!}{\Gamma(n+\kappa)}},
\end{equation}
where the normalization follows from the Laguerre orthogonality relation~\cite[Ch.~18]{DLMF}
\begin{equation}\label{eq:Laguerre-orth}
  \int_{0}^{\infty}z^{\alpha}\,e^{-z}\,L_{m}^{(\alpha)}(z)\,L_{n}^{(\alpha)}(z)\,\mathrm{d}z
  \;=\;\frac{\Gamma(n+\alpha+1)}{n!}\,\delta_{mn},
  \qquad\alpha>-1.
\end{equation}
The \textbf{bound subspace} $\Sigmah_{-}$ is spanned by the orthonormal
set $\{\psi_{n}\}_{n=0}^{\infty}$.

\section{Quantum symmetry regularization}\label{sec:quantum-reg}

The classical regularization reviewed in Section~\ref{sec:classical-reg} suggests its quantum analog:
for each sign $\pm$, a unitary isomorphism
$\hat\iota_{\pm}\colon\Sigmah_{\pm}\xrightarrow{\;\sim\;}\OOh_{\pm}$
under which the positive operator $|2H_{\kappa}|^{-1/2}$ corresponds to
the action of an element in the complexified Lie algebra $\mathfrak{sl}(2,\C)$ on the target $\OOh_{\pm}$.
Since quantization is not a functor, there is a priori no guarantee that such quantum analogues exist, nor is it clear what the relevant Lie algebra elements should be.
This question---which we call \emph{quantum (symmetry) regularization}---can be posed independently of the classical picture.
Nevertheless, the classical picture sharpens the question: it suggests that the targets should be
$\OOh_{-}=\D^{+}_{\kappa}$ and $\OOh_{+}=\D^{+}_{(\kappa+1)/2}$, an identification we adopt throughout this paper.

With the facts reviewed in Section~\ref{ssec:four-models} and Section~\ref{ssec:spectral-decomp}, this question admits an explicit answer:

\begin{theorem}[Quantum regularization]\label{thm:iotahat}
Realize both targets $\OOh_{\pm}$ on the common Whittaker half-line Hilbert space
$L^{2}(\R_{>0},t^{-1}\,\mathrm{d}t)$, equipped respectively with the representations
$\hat\pi_{\kappa}$ for $\OOh_{-}=\D^{+}_{\kappa}$ and $\hat\pi_{(\kappa+1)/2}$
for $\OOh_{+}=\D^{+}_{(\kappa+1)/2}$. Let $E,F$ denote the standard
$\sltwo$-basis elements of Section~\ref{ssec:four-models}.

\smallskip\noindent\textbf{(i) Negative-energy case.}
Set $h=\ii(F-E)$. The map $\psi_{n}\mapsto\hat\phi_{n}$, where
$\{\psi_{n}\}$ is the bound-state eigenbasis of $H_{\kappa}$ on $\Sigmah_{-}$
and $\{\hat\phi_{n}\}$ is the orthonormal $\Kt$-eigenbasis
of~\eqref{eq:hatphin-eigen}, extends by linearity and continuity to a unitary
isomorphism
\[
  \hat\iota_{-}\colon\Sigmah_{-}\xrightarrow{\;\sim\;}\OOh_{-}
\]
satisfying $\hat\iota_{-}^{-1}\,\hat h_{-}\,\hat\iota_{-}=H_{\kappa}$, where
$\hat h_{-}:=-1/\bigl(2(\hat\pi_{\kappa}(h/2))^{2}\bigr)$.

\smallskip\noindent\textbf{(ii) Positive-energy case.}
Let $\psi_{\lambda}:=\lambda^{-1}\,\psi_{E}\big|_{E=1/(2\lambda^{2})}$ denote the
scattering generalized eigenfunction~\eqref{eq:psi-E-Coulomb} reparametrized by
$\lambda=(2E)^{-1/2}>0$, so that
$\int_{0}^{\infty}\overline{\psi_{\lambda}(q)}\,\psi_{\lambda'}(q)\,\mathrm{d}q
=\delta(\lambda-\lambda')$. The map $\psi_{\lambda}\mapsto\sqrt{t}\,\delta(t-\lambda)$
extends to a unitary isomorphism
\[
  \hat\iota_{+}\colon\Sigmah_{+}\xrightarrow{\;\sim\;}\OOh_{+}
\]
satisfying $\hat\iota_{+}^{-1}\,\hat h_{+}\,\hat\iota_{+}=H_{\kappa}$, where
$\hat h_{+}:=1/\bigl(2(-\ii\hat\pi_{(\kappa+1)/2}(E))^{2}\bigr)$.
\end{theorem}

\begin{proof}
The argument in both cases rests on the same principle: two self-adjoint
operators on separable Hilbert spaces that have the same spectral type
(i.e.\ the same spectrum with the same multiplicity) are unitarily
equivalent, and matching (generalized) eigenfunctions with equal
eigenvalues defines the intertwining unitary isomorphism.

\smallskip\noindent\textbf{(i)}
On $\Sigmah_{-}$, the operator
$\hat N_{-}=(-2H_{\kappa})^{-1/2}$ is positive and self-adjoint with
purely discrete, simple spectrum
$\{n+\kappa/2\}_{n=0}^{\infty}$ and orthonormal eigenbasis
$\{\psi_{n}\}$.
On $\OOh_{-}=\D^{+}_{\kappa}$ (Whittaker model), the elliptic Cartan
generator $\hat\pi_{\kappa}(h/2)$ is self-adjoint with the same purely
discrete, simple spectrum $\{n+\kappa/2\}_{n=0}^{\infty}$ and
orthonormal eigenbasis $\{\hat\phi_{n}\}$
(cf.~\eqref{eq:hatphin-eigen}).
Since both eigenbases are complete orthonormal systems indexed by the
same eigenvalues, the map $\psi_{n}\mapsto\hat\phi_{n}$ extends by
linearity and continuity to a unitary isomorphism
$\hat\iota_{-}\colon\Sigmah_{-}\xrightarrow{\sim}\OOh_{-}$
satisfying $\hat\iota_{-}^{-1}\,\hat\pi_{\kappa}(h/2)\,\hat\iota_{-}
=\hat N_{-}$.
Squaring and taking reciprocals gives
$\hat\iota_{-}^{-1}\,\hat h_{-}\,\hat\iota_{-}
=-1/(2\hat N_{-}^{2})=H_{\kappa}$, where
$\hat h_{-}=-1/\bigl(2(\hat\pi_{\kappa}(h/2))^{2}\bigr)$.

\smallskip\noindent\textbf{(ii)}
On $\Sigmah_{+}$, the operator
$\hat N_{+}=(2H_{\kappa})^{-1/2}$ is positive and self-adjoint with
purely continuous, simple spectrum $(0,\infty)$ and Dirac-normalized
generalized eigenfunctions $\{\psi_{\lambda}\}_{\lambda>0}$.
On $\OOh_{+}=\D^{+}_{(\kappa+1)/2}$ (Whittaker model), the self-adjoint
companion $-\ii\hat\pi_{(\kappa+1)/2}(E)$ acts as
multiplication by $t>0$, hence also has purely continuous, simple
spectrum $(0,\infty)$ with Dirac-normalized generalized eigenfunctions
$\{\sqrt{t}\,\delta(t-\lambda)\}_{\lambda>0}$.
Since both operators have the same spectral type, matching generalized
eigenfunctions at each $\lambda$ defines the unitary isomorphism
$\hat\iota_{+}\colon\Sigmah_{+}\xrightarrow{\sim}\OOh_{+}$
satisfying $\hat\iota_{+}^{-1}\,(-\ii\hat\pi_{(\kappa+1)/2}(E))\,\hat\iota_{+}
=\hat N_{+}$.
Squaring and taking reciprocals gives
$\hat\iota_{+}^{-1}\,\hat h_{+}\,\hat\iota_{+}
=1/(2\hat N_{+}^{2})=H_{\kappa}$, where
$\hat h_{+}=1/\bigl(2(-\ii\hat\pi_{(\kappa+1)/2}(E))^{2}\bigr)$.
\end{proof}

This theorem is the quantum counterpart of Theorems~1 and~2
of~\cite{MMX2025}. The classical statements that $\iota_{\pm}$ are
symplectic embeddings with \emph{open dense} image are sharpened, on
quantization, to the statement that $\hat\iota_{\pm}$ are unitary
isomorphisms with image equal to \emph{all} of $\OOh_{\pm}$.
Equivalently, the one-parameter unitary group
$e^{-\ii\tau H_{\kappa}}$ ($\tau$ denoting time) on $\Sigmah_{\pm}$
is conjugate, via $\hat\iota_{\pm}$, to
$e^{-\ii\tau\hat h_{\pm}}$ on $\OOh_{\pm}$.

\section{Concluding remarks}\label{sec:concluding}

On the classical side, Ma--Meng--Xiao~\cite{MMX2025} realize each generalized
Kepler phase-space half $\Sigma_{\pm}$ as a dense open part of a
coadjoint orbit $\OO_{\pm}$ of $\sltwo$, and rename ``regularization'' as
\emph{S-duality}: an explicit change of variables that turns a hidden
symmetry into a manifest one, with the same flavor as the Fourier transform
between conjugate variables.

Quantizing this picture produces a strict quantum analog
$\hat\iota_{\pm}\colon\Sigmah_{\pm}\to\OOh_{\pm}$
(Theorem~\ref{thm:iotahat}): $\hat\iota_{-}$ identifies the
bound-state subspace with $\D^{+}_{\kappa}$, while $\hat\iota_{+}$
identifies the scattering subspace with $\D^{+}_{(\kappa+1)/2}$. Their pullback identities are the precise quantum
counterparts of the classical pullbacks $\iota_{\pm}^{*}h_{\pm}=H$.

Two structural lessons stand out.
\begin{description}[leftmargin=1em,style=nextline]
  \item[(a) Quantum S-duality parallels the classical decomposition.]
    Just as classically $\Sigma$ splits into $\Sigma_{-}\sqcup\Sigma_{+}$,
    the quantum Hilbert space splits into $\Sigmah_{-}\oplus\Sigmah_{+}$,
    and the two pieces are mapped to \emph{different} unitary
    representations:
    $\OOh_{-}=\D^{+}_{\kappa}$ is the unitary lowest-weight representation onto
    which $\Sigmah_{-}$ (where $H_{\kappa}$ has discrete spectrum) is mapped;
    $\OOh_{+}=\D^{+}_{(\kappa+1)/2}$ is the unitary lowest-weight representation
    onto which $\Sigmah_{+}$ (where $H_{\kappa}$ has continuous spectrum) is mapped.
    The distinction between scattering and bound is purely
    representation-theoretic: it lies in the choice of Lie algebra
    element---a nilpotent generator ($-\ii E$) for the positive-energy case
    and an elliptic generator ($h/2$) for the negative-energy case.
  \item[(b) Quantum construction is easier than classical.]
    The bulk of the work lies in formulating the problem in the right
    framework and assembling the requisite background from representation
    theory, special functions, and spectral theory.
    Once this is in place, the actual construction of
    $\hat\iota_{\pm}$ becomes transparent---it reduces to matching
    complete families of (generalized) eigenfunctions, with no
    hard techniques involved.
    The classical construction of $\iota_{\pm}$, by contrast, is substantially harder.
\end{description}

The present analysis fits into a lineage that begins with
Fock~\cite{Fock1935}, who, at each fixed negative energy level of the
3D hydrogen atom, used momentum-space stereographic projection to
$S^{3}$ to exhibit concretely the hidden $\SO(4)$ symmetry of that energy
eigenspace. Bander and
Itzykson~\cite{BanderItzykson1966I} subsequently
unified these per-level $\SO(n+1)$ actions (for the Coulomb problem in dimension $n\ge 2$) into a single irreducible
representation $D_{-}$ of the larger group $\SO_0(n+1,1)$, organizing all
bound levels into one Hilbert space. Indeed, they construct a unitary isomorphism
$\hat\iota_{-}\colon\Sigmah_{-}\to L^{2}(S^{n})$, where $L^{2}(S^{n})$ realizes the
representation $D_{-}$. Incidentally, the representation $\D^{+}_{\kappa}$ can be realized as a closed subspace of $L^{2}(S^{1})$ only when $\kappa=1$, in which case it is the Hardy space (the non-negative Fourier modes)~\cite{Knapp1986}. This strongly suggests that the 1D analogue of the hydrogen atom is the system with Hamiltonian $H_{1}$ (i.e.\ $\kappa=1$, equivalently $\mu=0$).

A natural next step is to repeat this analysis for the quantum
MICZ-Kepler problems and their higher-dimensional
analogs~\cite{Meng2007}. In~\cite{MengZhang2011}, Meng and Zhang have already constructed the analogue of $\hat\iota_{-}$, for which the target is realized as the space of $L^{2}$ sections of an appropriate vector bundle over the punctured Euclidean space of dimension $n$; it remains to work out the story for the scattering side, for which the group $\Spin_0(n,2)$ plays the role of $\Spin_0(n+1,1)$, following the guidance of the classical regularization results in~\cite{Meng2023}.


\bigskip
\noindent\textbf{Data availability} Data sharing is not applicable to this article, as no datasets were generated or analysed during the present study.

\section*{Declarations}
\noindent\textbf{Conflict of interest} The authors declare that they have no conflict of interest.



\begin{thebibliography}{99}

\bibitem{BanderItzykson1966I}
M.~Bander and C.~Itzykson,
\emph{Group theory and the hydrogen atom (I)},
Rev.\ Mod.\ Phys.\ \textbf{38} (1966) 330--345.

\bibitem{Fock1935}
V.~Fock,
\emph{Zur Theorie des Wasserstoffatoms},
Z.\ Phys.\ \textbf{98} (1935) 145--154.

\bibitem{HoweTan1992}
R.~Howe and E.-C.~Tan,
\emph{Non-Abelian Harmonic Analysis: Applications of $SL(2,\R)$},
Universitext, Springer, 1992.

\bibitem{Kirillov2004}
A.~A.~Kirillov,
\emph{Lectures on the Orbit Method},
Graduate Studies in Mathematics \textbf{64}, Amer.\ Math.\ Soc., Providence, RI, 2004.

\bibitem{Knapp1986}
A.~W.~Knapp,
\emph{Representation Theory of Semisimple Groups: An Overview Based on Examples},
Princeton Mathematical Series \textbf{36}, Princeton University Press, Princeton, NJ, 1986.

\bibitem{KustaanheimoStiefel1965}
P.~Kustaanheimo and E.~Stiefel,
\emph{Perturbation theory of Kepler motion based on spinor regularization},
J.~Reine Angew.\ Math.\ \textbf{218} (1965) 204--219.

\bibitem{LigonSchaaf1976}
T.~Ligon and M.~Schaaf,
\emph{On the global symmetry of the classical Kepler problem},
Rep.\ Math.\ Phys.\ \textbf{9} (1976) 281--300.

\bibitem{MMX2025}
J.~Ma, G.~Meng, J.~Xiao,
\emph{The symmetry regularization of 1D generalized Kepler problems},
J.~Geom.\ Phys.\ \textbf{216} (2025) 105576.

\bibitem{Meng2007}
G.~Meng,
\emph{MICZ-Kepler problems in all dimensions},
J.~Math.\ Phys.\ \textbf{48} (2007) 032105.

\bibitem{Meng2014}
G.~Meng,
\emph{The universal Kepler problem},
J.~Geom.\ Symmetry Phys.\ \textbf{36} (2014) 47--57.

\bibitem{Meng2023}
G.~Meng,
\emph{Ligon--Schaaf regularization map revisited},
J.~Math.\ Phys.\ \textbf{64} (2023) 122902.

\bibitem{Meng2024Oziewicz}
G.~Meng,
\emph{Kepler problem, Lorentz transformation, and Jordan algebra},
Chap.~15 (pp.~309--322) in
H.~M.~Colin Garcia, J.~de~J.~Cruz Guzm\'an, L.~H.~Kauffman, H.~Makaruk (eds.),
\emph{Scientific Legacy of Professor Zbigniew Oziewicz: Selected Papers from the
International Conference ``Applied Category Theory Graph-Operad-Logic''},
Series on Knots and Everything, World Scientific, 2024.

\bibitem{MengZhang2011}
G.~Meng and R.~Zhang,
\emph{Generalized MICZ-Kepler problems and unitary highest weight modules},
J.~Math.\ Phys.\ \textbf{52} (2011) 042106.

\bibitem{Moser1970}
J.~Moser,
\emph{Regularization of Kepler's problem and the averaging method on a manifold},
Commun.\ Pure Appl.\ Math.\ \textbf{23} (1970) 609--636.

\bibitem{DLMF}
F.~W.~J.~Olver, A.~B.~Olde Daalhuis, D.~W.~Lozier, B.~I.~Schneider,
R.~F.~Boisvert, C.~W.~Clark, B.~R.~Miller, B.~V.~Saunders, H.~S.~Cohl,
and M.~A.~McClain, eds.,
\emph{NIST Digital Library of Mathematical Functions},
\url{https://dlmf.nist.gov/}, Release 1.2.6 of 2026-03-15.

\bibitem{Sally1967}
P.~J.~Sally, Jr.,
\emph{Analytic continuation of the irreducible unitary representations of the
universal covering group of $SL(2,\R)$},
Mem.\ Amer.\ Math.\ Soc.\ \textbf{69} (1967).

\bibitem{Sniatycki1980}
J.~\'{S}niatycki,
\emph{Geometric Quantization and Quantum Mechanics},
Applied Mathematical Sciences \textbf{30}, Springer, New York, 1980.

\bibitem{WhittakerWatson1927}
E.~T.~Whittaker and G.~N.~Watson,
\emph{A Course of Modern Analysis}, 4th ed.,
Cambridge University Press, 1927.

\bibitem{Woodhouse1991}
N.~M.~J.~Woodhouse,
\emph{Geometric Quantization},
2nd ed., Oxford Mathematical Monographs, Clarendon Press, Oxford, 1991.

\end{thebibliography}
\end{document}